\documentclass[11pt]{article}
\pdfoutput=1

\usepackage[utf8]{inputenc}
\usepackage[T1]{fontenc}
\usepackage{fullpage}
\usepackage{verbatim}
\usepackage[english]{babel}
\usepackage[colorlinks=true, allcolors=blue, linktoc=page]{hyperref}
\usepackage{amsmath,amsthm,amssymb,amscd,amsfonts,mathtools,mathrsfs,braket,bbm,enumerate,graphicx,xcolor,cancel}
\numberwithin{equation}{section}
\usepackage{tocloft}
\usepackage{tikz}
\usetikzlibrary{quantikz2}
\usepackage{algorithm,algpseudocode}

\allowdisplaybreaks
\usepackage{authblk}

\usepackage[para]{footmisc}

\theoremstyle{definition}
\newtheorem{theorem}{Theorem}[section]
\newtheorem{proposition}[theorem]{Proposition}
\newtheorem{lemma}[theorem]{Lemma}

\newtheorem{fact}[theorem]{Fact}
\newtheorem{remark}[theorem]{Remark}
\newtheorem{example}[theorem]{Example}
\newtheorem{definition}{Definition}

\def\be{\begin{eqnarray}}
\def\ee{\end{eqnarray}}
\def\bex{\begin{eqnarray*}}
\def\eex{\end{eqnarray*}}

\DeclareMathOperator{\tr}{Tr}

\DeclareMathOperator{\re}{Re}
\DeclareMathOperator{\im}{Im}

\definecolor{Pr}{rgb}{0.4,0.3,0.9}

\title{A short tutorial on Wirtinger Calculus with applications in quantum information}

\author[1]{Kelvin Koor\thanks{\tt cqtkjkk@nus.edu.sg}}
\author[1]{Yixian Qiu\thanks{\tt yixian\_qiu@u.nus.edu}}
\author[1,2,3]{Leong Chuan Kwek\thanks{\tt kwekleongchuan@nus.edu.sg}}
\author[1]{Patrick Rebentrost\thanks{\tt cqtfpr@nus.edu.sg}}
\affil[1]{Centre for Quantum Technologies, National University of Singapore}
\affil[2]{MajuLab, CNRS-UNS-NUS-NTU International Joint Research Unit, UMI 3654}
\affil[3]{Quantum Science and Engineering Centre, Nanyang Technological University}
\date{\vspace{-5ex}}

\begin{document}
\maketitle

\begin{abstract}
The optimization of system parameters is a ubiquitous problem in science and engineering. The traditional approach involves setting to zero the partial derivatives of the objective function with respect to each parameter, in order to extract the optimal solution. However, the system parameters often take the form of complex matrices. In such situations, conventional methods become unwieldy. The `Wirtinger Calculus' provides a relatively simple methodology for such optimization problems. In this tutorial, we provide a pedagogical introduction to Wirtinger Calculus. To illustrate the utility of this framework in quantum information theory, we also discuss a few example applications.
\end{abstract}

\tableofcontents

\section{Introduction}
The optimization of system parameters is a ubiquitous problem in science and engineering. For an analytical solution, the partial derivatives of the objective function with respect to the parameters are set to zero. Alternatively, gradient descent could be applied iteratively, leading to increasingly accurate approximations to the optimal solution. In either case, this involves taking the (partial) derivatives of the objective function. Often, however, the problem is formulated in terms of \textit{complex}-valued parameters. While the complex derivatives of complex functions play an important role in complex analysis and its numerous applications to science and engineering, in the optimization of \textit{real}-valued\footnote{Objective functions are real- and not complex-valued because there is no ordering ($<$ and $>$) in the field of complex numbers, and thus no notion of minimization/maximization.} functions of \textit{complex} parameters they are much less useful. Furthermore, these parameters often come in the form of matrices, whose structure we would like to preserve for a deeper understanding of the problem.

In short, we are dealing with real-valued functions of complex matrices. How do we optimize such functions? In principle we could convert everything into real numbers---$M_N(\mathbb{C}) \cong \mathbb{R}^{2N^2}$, so we could view $f$ as a function of $2N^2$ real parameters and implement conventional optimization methods. However, this conversion is generally tedious, and the resulting expression for $f$ cumbersome. The `Wirtinger Calculus' provides a relatively simple methodology for the optimization of such functions, through the use of `Wirtinger derivatives'. This calculus is best viewed as a bookkeeping device. It justifies, and enables us to `differentiate as usual' (almost) with respect to the complex matrix parameter as a whole, as opposed to $2N^2$ distinct real parameters. Any dependencies among the elements of the input matrix (e.g. when the matrix is Hermitian) can also be accounted for via appropriate modifications.

The Wirtinger Calculus was developed principally by Austrian mathematician Wilhelm Wirtinger \cite{wirtinger1927formalen} in his paper on functions of several complex variables, although according to \cite{remmert1991theory} the concept goes back to Poincar\'{e} \cite{poincare1898proprietes}. It was subsequently rediscovered and extended by the electrical engineering community \cite{brandwood1983complex,van1994complex} for the purposes of optimizing real-valued functions with complex inputs. A compendium on this topic can be found in Hj\o{}rungnes' text \cite{hjorungnes2011complex}, alongside a concise summary \cite{hjorungnes2007complex}. Other good resources include \cite{kreutz2009complex}, Chapter 1 (Part 4) of \cite{remmert1991theory} and \cite{greene2006function}. Our exposition in this article is elementary, and only requires an understanding of the basics of multivariable calculus, linear algebra and complex analysis.

\subsection{Motivation}
The objects of interest in this discussion are real-valued functions with (generally complex) matrix inputs. We begin with a simple example that is not entirely trivial. Consider the function
\begin{align}\label{Equation:motivating_f}
    f: M_2(\mathbb{R}) &\longrightarrow \mathbb{R}\\\nonumber
    \mathbf{X} &\mapsto \tr(\mathbf{X}^2).
\end{align}
How do we optimize this function? The obvious way to do so would be to observe that for the purposes of optimization, the matrix structure of $\mathbf{X}$ can be disregarded. What matters is that $\mathbf{X}$ essentially comprises a set of parameters, and $f$ is a multivariable function. If we write 
$\mathbf{X}=
\begin{bmatrix}
    a & b\\
    c & d
\end{bmatrix}
$, then after identifying $\mathbf{X}$ with the vector $(a,b,c,d)$ we can express\footnote{Strictly speaking, the functions in Equations \ref{Equation:motivating_f} and \ref{Equation:motivating_new_f} are \textit{not} the same, simply because they have different domains ($M_2(\mathbb{R})$ versus $\mathbb{R}^4$).} $f$ as
\begin{align}\label{Equation:motivating_new_f}
    f(a,b,c,d) = a^2 + d^2 + 2bc.
\end{align}
Conventional optimization methods can then be implemented accordingly. (Note that constraints may have to be further imposed on $\mathbf{X}$ to ensure the existence of solutions. For example in Equation \ref{Equation:motivating_f} above we note that simply by rescaling $\mathbf{X}$ we can send $f(\mathbf{X})$ to $\pm \infty$, thus no optimal solution exists. To remedy this we could impose, say, $\tr \mathbf{X}=1$.)

One quickly observes the impracticality of this approach except in the simplest of cases. Even if the form of $f$ is relatively simple, like the one above, a few issues persist. For instance, if $\mathbf{X} \in M_n(\mathbb{R})$ for general $n$, then as $n$ increases the corresponding analog to $f$ in Equation \ref{Equation:motivating_new_f} becomes considerably messier. Likewise, had the field been $\mathbb{C}$ instead of $\mathbb{R}$, or if there were constraints among the elements of $\mathbf{X}$ (if $\mathbf{X} \in M_n(\mathbb{R})$ is symmetric, say, then the number of Lagrange multipliers required goes as $\Theta(n^2)$). But perhaps the most salient deficiency lies beyond the mathematics: that $\mathbf{X}$ is a matrix (and not just a bunch of numbers) often holds significant physical meaning, and one would like to preserve this structure to gain further insight into the problem. This is perhaps best illustrated with an example. Suppose the optimal solution to some problem (similar in flavour to Problem \ref{Equation:motivating_f} above) is given by
$\mathbf{X^\star}=
\begin{bmatrix}
    a^\star & b^\star\\
    c^\star & d^\star
\end{bmatrix}
$ with
\begin{align*}
    a^\star &= p^3 - p^2 + 2pqr + qrs - qr\\
    b^\star &= p^2q + q^2r + pqs + qs^2 - pq - qs\\
    c^\star &= p^2r + qr^2 + prs + rs^2 - pr - ps\\
    d^\star &= s^3 - s^2 + 2qrs + pqr - qr
\end{align*}
where 
$\mathbf{Y}= 
\begin{bmatrix}
    p & q\\
    r & s
\end{bmatrix}
$ is a matrix involved in the specification of the problem. One suspects it might be possible to express $\mathbf{X^\star}$ in terms of $\mathbf{Y}$, but it is not immediately obvious how\footnote{Here, $\mathbf{X^\star} = \mathbf{Y}^3 - \mathbf{Y}^2$. The solution form is admittedly rather artificial. For a `realistic' example see Example \ref{Example:constrained_vNEntropy} below.}. By using the `obvious' method above, we have obscured potentially important information in the form of an explicit connection between $\mathbf{X^\star}$ and $\mathbf{Y}$. 

Wouldn't it be nice if we could somehow differentiate $f(\mathbf{X})$ with respect to $\mathbf{X}$ as a whole? That is, we obtain something like
\begin{align*}
    \frac{d f(\mathbf{X})}{d \mathbf{X}} = f'(\mathbf{X})
\end{align*}
for some matrix function $f'$, which we can call the matrix derivative of $f$. Then by setting $f'(\mathbf{X})=0$ we could hopefully extract the optimal $\mathbf{X}$, in a similar flavour to the procedure taught in elementary calculus. Showing that this could be done for a large class of interesting and relevant functions, and how to do it, is the focus of this article.

\subsection{Preliminaries and notation}
We define the following notations. Let $\mathbb{N} = \{1,2,\dots\}$ be the set of positive natural numbers. For $n \in \mathbb{N}$, $[n] = \{1,2,\dots,n\}$. If $z \in \mathbb{C}$, $z^* \in \mathbb{C}$ denotes its complex conjugate. $M_n(\mathbb{F})/\mathbb{F}^{n \times n}$ is the set of $n \times n$ dimensional matrices over the field $\mathbb{F}$, where $\mathbb{F} = \mathbb{R}/\mathbb{C}$. We denote a Hilbert space by $\mathcal{H}$ ($\mathcal{H}_N$ if its dimension is to be explicitly specified), the set of linear operators on $\mathcal{H}$ by $\mathcal{L}(\mathcal{H})$, and the set of density operators on $\mathcal{H}$ by $\mathcal{D}(\mathcal{H})$. $\|\cdot\|_F$ is the Frobenius norm. The symbol $\odot$ denotes component-wise product, e.g. for vectors $(v \odot w)_i = v_iw_i$, for matrices $(A \odot B)_{ij} = A_{ij}B_{ij}$. Throughout most of this text, we adopt the convention whereby scalar quantities are denoted by lowercase symbols ($x$), vector quantities by lowercase boldface $(\mathbf{x})$ and matrix quantities by capital boldface $(\mathbf{X})$.

We will not be pedantic with domain/codomain issues. While we generally write $f: \mathbb{C} \longrightarrow \mathbb{C}$ or $f: \mathbb{R}^n \longrightarrow \mathbb{R}^m$, points where $f$ is not well-defined are implicitly excluded from the domain. E.g. if $f(z) = 1/z$ then the domain of $f$ is understood to be $\mathbb{C}\setminus\{0\}$. Similarly, a function is \textit{complex-differentiable} at a point if its complex derivative exists at that point, and it is \textit{holomorphic} on an open set if it is complex-differentiable at every point in the open set. This is strictly speaking a stronger condition than complex-differentiability, but we shall simply use the two terms interchangeably.

\subsection{Real vs. Complex Differentiability}\label{Section:real_vs_complex_differentiability}
As will be shown below in Proposition \ref{Proposition:Optimization_with_Wirtinger_derivatives}, non-constant real-valued functions of complex numbers cannot be complex-differentiable, in the sense that the quantity
\begin{align*}
\frac{df(z)}{dz} = \lim_{h \rightarrow 0} \frac{f(z+h)-f(z)}{h}
\end{align*}
is not well-defined. However, the optimization of $f$ entails taking the derivatives $\frac{\partial f}{\partial x}$ and $\frac{\partial f}{\partial y}$, which we assume a priori are well-defined. That $\frac{\partial f}{\partial x}$ and $\frac{\partial f}{\partial y}$ exist, but not $\frac{df}{dz}$, might cause some confusion since $z=(x,y)$. In this section, we clarify the difference between real- and complex- differentiability of a function $f:\mathbb{C} \longrightarrow \mathbb{C}$. The material presented here is standard fare in real and complex analysis, see for example \cite{rudin1953principles,remmert1991theory,greene2006function}.

In single-variable calculus, the derivative of $f: \mathbb{R} \longrightarrow \mathbb{R}$ at $x$ is
\begin{align}\label{Equation:real_derivative}
    f'(x) = \lim_{h \rightarrow 0} \frac{f(x+h)-f(x)}{h},
\end{align}
with the usual interpretation being the `slope of the tangent of $f$ at $x$'. For multivariable functions the generalization of Equation \ref{Equation:real_derivative} is as follows. We say $f: \mathbb{R}^n \longrightarrow \mathbb{R}^m$ is \textit{real-differentiable} at $x$ if there exists $T \in \mathcal{L}(\mathbb{R}^n,\mathbb{R}^m)$ such that
\begin{align}\label{Equation:real_derivative_multivariable}
    \lim_{h \rightarrow 0} \frac{\|f(x+h)-f(x)-T(h)\|}{\|h\|} = 0,
\end{align}
where $\|\cdot\|$ denotes the Euclidean norm. We write $f'(x) = T$ and call it the derivative/differential of $f$ at $x$. Thus, this generalized derivative is a linear operator, not just a number (when $n=m=1$, a linear operator is just a scaling, so it can be identified with the scale factor, a real number). The matrix representation of $f'(x)$ with respect to the standard Euclidean bases is the Jacobian matrix $\left[\frac{\partial f_i}{\partial x_j}(x)\right]_{\substack{1 \leq i \leq m\\1 \leq j \leq n}}$. While the geometric picture of a slope no longer holds, the interpretation of $f'(x)$ as best linear approximation of $f$ at $x$ still applies.
The differential also subsumes as special cases the directional derivative (when $n=1$) and the gradient vector (when $m=1$). It is a standard result that for an open set $O \subseteq \mathbb{R}^n$, $f$ is real-differentiable on $O$ if and only if all its partial derivatives $\frac{\partial f_i}{\partial x_j}$, $1 \leq i \leq m, 1 \leq j \leq n$ exist and are continuous on $O$ (Theorem 9.21, \cite{rudin1953principles}). Thus for functions of our interest, we can assume they are real-differentiable since we almost always assume $\frac{\partial f_i}{\partial x_j}$ are all smooth.

The complex derivative is obtained by directly extending the concept of infinitesimal difference quotients in calculus (Eq. \ref{Equation:real_derivative}) to the complex domain. A function $f: \mathbb{C} \longrightarrow \mathbb{C}$ is \textit{complex-differentiable} at $z$ if the derivative $f'(z) = \lim_{h \rightarrow 0} \frac{f(z+h)-f(z)}{h}$ exists, or equivalently,
\begin{align}\label{Equation:complex_derivative}
    \lim_{h \rightarrow 0} \frac{f(z+h)-f(z)-f'(z)h}{h} = 0.
\end{align}
Since $f: \mathbb{C} \longrightarrow \mathbb{C}$ is also $f: \mathbb{R}^2 \longrightarrow \mathbb{R}^2$, the notion of real-differentiability also applies to $f$. How do Equations \ref{Equation:real_derivative_multivariable} and \ref{Equation:complex_derivative} relate to each other? First, recall that the Cauchy-Riemann equations
\begin{align}\label{Equation:Cauchy-Riemann}
    \frac{\partial u}{\partial x} = \frac{\partial v}{\partial y}, \qquad \frac{\partial u}{\partial y} = -\frac{\partial v}{\partial x}
\end{align}
holding on $\mathcal{O}$ form a necessary and sufficient condition for $f$ to be complex-differentiable on $\mathcal{O}$. In this case, $f'(z) = u_x+iv_x = v_y-iu_y$. We next recast \ref{Equation:complex_derivative} into a form more similar to \ref{Equation:real_derivative_multivariable}. The multiplication of two complex numbers can be viewed as a matrix acting on a vector in $\mathbb{R}^2$: for $z,w \in \mathbb{C}$,
\begin{align*}
    z \cdot w = (\re z \re w - \im z \im w) + i(\re z \im w + \im z \re w) &=
    \begin{bmatrix}
        \re z \re w - \im z \im w\\
        \re z \im w + \im z \re w
    \end{bmatrix}\\
    &= 
    \begin{bmatrix}
        \re z & -\im z\\
        \im z & \re z
    \end{bmatrix}
    \begin{bmatrix}
        \re w\\
        \im w        
    \end{bmatrix}.
\end{align*}
Thus we have that if $f'(z)$ exists, then
\begin{alignat*}{5}
    &\lim_{h \rightarrow 0} \frac{f(z+h)-f(z)-f'(z)h}{h} = 0 &&\\
    \implies &\lim_{h \rightarrow 0} \frac{f(z+h)-f(z)-
    \begin{bmatrix}
        u_x & -v_x\\
        v_x & u_x
    \end{bmatrix}h
    }{h} = 0\\
    \implies &\lim_{h \rightarrow 0} \frac{f(z+h)-f(z)-
    \begin{bmatrix}
        u_x & u_y\\
        v_x & v_y
    \end{bmatrix}h
    }{h} = 0 \qquad &&\text{by Cauchy-Riemann}\\
    \implies &\lim_{h \rightarrow 0} \frac{f((x,y)+h)-f(x,y)-
    f'(x,y)h
    }{h} = 0 \qquad &&\text{since $f'(x,y) = 
    \begin{bmatrix}
        u_x & u_y\\
        v_x & v_y
    \end{bmatrix}$}\\
    \iff &\lim_{h \rightarrow 0} \frac{\|f((x,y)+h)-f(x,y)-
    f'(x,y)h
    \|}{\|h\|} = 0,
\end{alignat*}{5}
i.e. $f$ is real-differentiable, and the complex derivative $f'(z)$ (a complex number) is identified with the differential/real derivative $f'(x,y)$ (a linear operator). Furthermore, note that 
\begin{align*}
    f'(z) = f'(x,y) = 
    \begin{bmatrix}
        u_x & u_y\\
        v_x & v_y
    \end{bmatrix} = 
    \begin{bmatrix}
        u_x & -v_x\\
        v_x & u_x
    \end{bmatrix} = rA,
\end{align*}
where $0 \leq r \in \mathbb{R}$ and $A \in SO_2(\mathbb{R})$. This operator is a composition of scaling and an orientation-preserving rotation (the scaling factor and rotation angle generally differ from point to point). Such maps are widely called `conformal' in the literature. In short, we have the following fact.
\begin{fact}
    If $f: \mathbb{C} \longrightarrow \mathbb{C}$ is complex-differentiable/holomorphic, then it is real-differentiable. In this case, its differential $f'(x,y)$ is a conformal operator and can be identified with $f'(z)$.
\end{fact}
The converse does not hold: if a map $f$ is real-differentiable it need not be complex-differentiable, the classic example being complex conjugation: $f(z) = z^\star$ or equivalently $f(x,y) = (x,-y)$. The conformality of the differential makes complex-differentiability a much more rigid property compared to real-differentiability. This rigidity however underlies many remarkable properties of holomorphic functions and thus the power of complex analysis, broadly defined as the study of these functions.

\section{Wirtinger Calculus}

\subsection{The basic idea}\label{Subsection:Basic_idea}
We begin by discussing general complex-valued functions, $f: \mathbb{C} \longrightarrow \mathbb{C}$. Since $\mathbb{C}$ is just $\mathbb{R}^2$ endowed with the multiplication operation $(a,b) \times (c,d) \mapsto (ac-bd,ad+bc)$, we can view $f$ as
\begin{align*}
    f:\; &\mathbb{R}^2 \longrightarrow \mathbb{R}^2\\
    &(x,y) \mapsto (u(x,y),v(x,y)).
\end{align*}
Here we assume $f$ is real-differentiable, but $f$ need not be complex-differentiable (refer to Section \ref{Section:real_vs_complex_differentiability} above for the difference). Often, however, $f$ is expressed in terms of $z$ and $z^*$. More precisely, viewing $z,z^*$ as functions from $\mathbb{R} \times \mathbb{R}$ to $\mathbb{C}$, one has a function $\tilde{f}: \mathbb{C} \times \mathbb{C} \longrightarrow \mathbb{C}$ such that
\begin{align}\label{Equation:tildef_to_f}
    f(x,y) =: \tilde{f}\circ (z,z^*)(x,y) = \tilde{f}(z(x,y),z^*(x,y)) = \tilde{f}(x+iy,x-iy).
\end{align}
Conversely, we also have
\begin{align}\label{Equation:f_to_tildef}
    \tilde{f}(z,z^*) = f\circ (x,y)(z,z^*) = f(x(z,z^*),y(z,z^*)) = f\left(\frac{z+z^*}{2},\frac{z-z^*}{2}\right).
\end{align}
Partial differentiating $f$ and applying the chain rule gives
\begin{align*}
    \frac{\partial f}{\partial x}(x,y) &= \frac{\partial \tilde{f}}{\partial z}\frac{\partial z}{\partial x} + \frac{\partial \tilde{f}}{\partial z^*}\frac{\partial z^*}{\partial x} = \frac{\partial \tilde{f}}{\partial z}(z(x,y),z^*(x,y)) + \frac{\partial \tilde{f}}{\partial z^*}(z(x,y),z^*(x,y))\\
    \frac{\partial f}{\partial y}(x,y) &= \frac{\partial \tilde{f}}{\partial z}\frac{\partial z}{\partial y} + \frac{\partial \tilde{f}}{\partial z^*}\frac{\partial z^*}{\partial y} = i \cdot \frac{\partial \tilde{f}}{\partial z}(z(x,y),z^*(x,y)) - i \cdot \frac{\partial \tilde{f}}{\partial z^*}(z(x,y),z^*(x,y)).
\end{align*}
After rearranging terms,
\begin{align*}
    \frac{\partial \tilde{f}}{\partial z}(z(x,y),z^*(x,y)) &= \frac{1}{2}\left(\frac{\partial f}{\partial x} - i\frac{\partial f}{\partial y}\right) (x,y)\\
    \frac{\partial \tilde{f}}{\partial z^*}(z(x,y),z^*(x,y)) &= \frac{1}{2}\left(\frac{\partial f}{\partial x} + i\frac{\partial f}{\partial y}\right) (x,y).
\end{align*}
The notations $z,z^*$ may raise questions on independence. This is irrelevant---one may simply write $z_1,z_2$ if one wishes. We emphasize that the fundamental input variables are the two real numbers $x$ and $y$. 

The purpose of the discussion above is as follows. We have that $f$ and $\tilde{f}$ are two different expressions of the complex function $f: \mathbb{C} \longrightarrow \mathbb{C}$, $f$ being expressed in terms of $x$ and $y$, while $\tilde{f}$ being expressed in terms of $z$ and $z^*$. Often, it is the `tilde-form' of $f$ that is given. In such cases, certain procedures become a hassle. Examples include verifying the Cauchy-Riemann conditions, or finding the optimal points for $f$ (in the case $f$ is real-valued). The common feature of these operations is that they involve taking the partial derivatives $\frac{\partial}{\partial x}, \frac{\partial}{\partial y}$ of $f$, which requires $f$ to be expressed in terms of $x$ and $y$ beforehand. This conversion (from $\tilde{f}(z,z^*)$ to $f(x,y)$, using Eq. \ref{Equation:tildef_to_f}) is generally tedious, and the resulting form of $f(x,y)$ cumbersome. For instance, consider $\tilde{f}(z,z^*)=z^mz^{*n}$ for large $m,n \in \mathbb{Z}$. As we shall see below, the operators $\frac{\partial}{\partial z}, \frac{\partial}{\partial z^*}$ allow us to circumvent this by retaining the form $\tilde{f}(z,z^*)$ and partially differentiate with respect to $z$ and $z^*$ instead. In a nutshell, we can view $\tilde{f}(z,z^*)$ as a black box encapsulating the complexity of $f(x,y)$ and $\frac{\partial}{\partial z},\frac{\partial}{\partial z^*}$ as `higher-level' operators acting on the black box.

\begin{definition}[Wirtinger Derivatives]
    Let $f(x,y) = (u(x,y),v(x,y))$ be a complex function where $u(x,y)$ and $v(x,y)$ are differentiable functions with respect to $x,y$.\footnote{We often just assume $u$ and $v$ are smooth.} Write $f(x,y)=\tilde{f}(z(x,y),z^*(x,y))$, where $\tilde{f}: \mathbb{C} \times \mathbb{C} \longrightarrow \mathbb{C}$. The Wirtinger derivatives of $f$ are the partial derivatives of $\tilde{f}$ with respect to $z$ and $z^*$:
    \begin{align}
        \frac{\partial \tilde{f}}{\partial z} \qquad \text{and} \qquad \frac{\partial \tilde{f}}{\partial z^*}.
    \end{align}
\end{definition}

\begin{remark}
    As discussed above, the Wirtinger derivatives satisfy
    \begin{align}\label{Equation:Wirtinger_interpretation}
        \frac{\partial \tilde{f}}{\partial z}(z(x,y),z^*(x,y)) &= \frac{1}{2}\left(\frac{\partial f}{\partial x} - i\frac{\partial f}{\partial y}\right) (x,y)\\\nonumber
        \frac{\partial \tilde{f}}{\partial z^*}(z(x,y),z^*(x,y)) &= \frac{1}{2}\left(\frac{\partial f}{\partial x} + i\frac{\partial f}{\partial y}\right) (x,y).
    \end{align}
    These two equations tell us what is happening `under the hood': partially differentiating with respect to $z/z^*$ then writing $z,z^*$ in terms of $x,y$ is equivalent to first converting $\tilde{f}(z,z^*)$ into $f(x,y)$ then applying the operators $\frac{1}{2}\left(\frac{\partial f}{\partial x} \mp i\frac{\partial f}{\partial y}\right)$. Often, we abuse notation and disregard the distinction between $\tilde{f}$ and $f$ (i.e. write both $f(x,y)$ and $f(z,z^*)$). Doing so, we can simply write
    \begin{align}
        \frac{\partial}{\partial z} = \frac{1}{2}\left(\frac{\partial}{\partial x} - i\frac{\partial}{\partial y}\right), \qquad \frac{\partial}{\partial z^*} = \frac{1}{2}\left(\frac{\partial}{\partial x} + i\frac{\partial}{\partial y}\right).
    \end{align}
    \underline{Henceforth we shall abuse notation and do away with $\tilde{f}$, writing both $f(x,y)$ and $f(z,z^*)$} (simply bear in mind that $f$ really is a function of two real variables, not two complex variables). Thus, $\frac{\partial}{\partial z}f(z,z^*)$ really means $\frac{\partial}{\partial z}\tilde{f}(z,z^*)$ and $\frac{\partial}{\partial z}f(x,y)$ really means $\frac{1}{2}\left( \frac{\partial}{\partial x} - i\frac{\partial}{\partial y} \right)f(x,y)$ (likewise for $\frac{\partial}{\partial z^*}$).
\end{remark}

\begin{example}
    Let $f(z,z^*)=z$. Then $\frac{\partial}{\partial z}z = 1 = \frac{1}{2}\left( \frac{\partial}{\partial x} - i\frac{\partial}{\partial y} \right)(x+iy)$ and $\frac{\partial}{\partial z^*}z = 0 = \frac{1}{2}\left( \frac{\partial}{\partial x} + i\frac{\partial}{\partial y} \right)(x+iy)$. Likewise, $\frac{\partial}{\partial z}z^* = 0 = \frac{1}{2}\left( \frac{\partial}{\partial x} - i\frac{\partial}{\partial y} \right)(x-iy)$ and $\frac{\partial}{\partial z^*}z^* = 1 = \frac{1}{2}\left( \frac{\partial}{\partial x} + i\frac{\partial}{\partial y} \right)(x-iy)$. 
\end{example}

\begin{example}
Consider the function $f(z,z^*) = z^2z^*$. Here $f(x,y) = (x+iy)^2(x-iy) = u(x,y) + iv(x,y)$ where $u(x,y) = x^3+xy^2$, $v(x,y) = x^2y+y^3$. We verify Equations \ref{Equation:Wirtinger_interpretation}:
\begin{alignat*}{5}
    \frac{\partial f}{\partial z} &= 2zz^* = 2(x^2+y^2) &&= \frac{1}{2}\left( \frac{\partial}{\partial x} - i\frac{\partial}{\partial y} \right)(u+iv)\\
    \frac{\partial f}{\partial z^*} &= z^2 = x^2+2ixy+y^2 &&= \frac{1}{2}\left( \frac{\partial}{\partial x} + i\frac{\partial}{\partial y} \right)(u+iv).
\end{alignat*}
\end{example}

We state without proof the following:
\begin{fact}
    The Wirtinger derivatives satisfy the following properties:
    \begin{enumerate}[i.]
        \item Linearity: $\frac{\partial}{\partial z/z^*}(af+bg) = a\frac{\partial f}{\partial z/z^*}+b\frac{\partial g}{\partial z/z^*}$ \quad \text{for} $a,b \in \mathbb{C}$
        \item Product: $\frac{\partial}{\partial z/z^*}(f \cdot g) = \frac{\partial f}{\partial z/z^*}\cdot g + f\cdot \frac{\partial g}{\partial z/z^*}$
        \item Chain: $\frac{\partial}{\partial z/z^*}(g \circ f) = \frac{\partial g}{\partial f}\frac{\partial f}{\partial z/z^*} + \frac{\partial g}{\partial f^*}\frac{\partial f^*}{\partial z/z^*}$
        \item Conjugate: $\frac{\partial f}{\partial z^*} = \left( \frac{\partial f^*}{\partial z} \right)^*$. In particular, if $f$ is real-valued then $\frac{\partial f}{\partial z^*} = \left( \frac{\partial f}{\partial z} \right)^*$.
    \end{enumerate}
\end{fact}

We emphasize that $f$ need not be complex-differentiable---so long as $f(x,y)$ is real-differentiable, or equivalently, the partial derivatives $u_x,u_y,v_x,v_y$ exist (see Section \ref{Section:real_vs_complex_differentiability}), the Wirtinger derivatives of $f$ exist. If $f$ is complex-differentiable, however, one Wirtinger derivative becomes zero and the other reduces to the normal complex derivative.
\begin{proposition}\label{Proposition:Wirtinger_derivatives_when_f_holomorphic}
    Let $\mathcal{O} \subseteq \mathbb{C}$ be an open set. $f$ is holomorphic on $\mathcal{O}$ if and only if the following hold on $\mathcal{O}$:
    \begin{alignat*}{5}
        \frac{\partial f}{\partial z^*} &= 0 \qquad &&(\text{Cauchy-Riemann condition})\\
        \frac{\partial f}{\partial z} &= \frac{df}{dz} \qquad &&(\text{reduction to usual complex derivative}).
    \end{alignat*}
\end{proposition}
\begin{proof}
    On an open set, $f$ is holomorphic if and only if $u_x=v_y$, $u_y=-v_x$ (the Cauchy-Riemann conditions) hold, see Section \ref{Section:real_vs_complex_differentiability}. Therefore
    \begin{align*}
        \frac{\partial f}{\partial z^*} = \frac{1}{2}(f_x+if_y) = \frac{1}{2}(u_x+iv_x+iu_y-v_y) = 0
    \end{align*}
    and
    \begin{align*}
        \frac{\partial f}{\partial z} = \frac{1}{2}(f_x+if_y) = \frac{1}{2}(u_x+iv_x-iu_y+v_y) = u_x+iv_x = f'(z).
    \end{align*}
\end{proof}

The concise equation $\frac{\partial f}{\partial z^*}=0$ thus merges the two Cauchy-Riemann equations into one. This also agrees with conventional wisdom whereby complex-differentiable functions `have no $z^*$ terms in their expressions'. The next result enables us to optimize for real-valued $f$ using Wirtinger derivatives.
\begin{proposition}\label{Proposition:Optimization_with_Wirtinger_derivatives}
    Let $f: \mathbb{C} \longrightarrow \mathbb{R}$ be a real-valued function. Then 
    \begin{enumerate}[i.]
        \item If $f$ is nonconstant, then it cannot be holomorphic, i.e. $\frac{df}{dz}$ does not exist.
        \item $f$ has a stationary point at $z=(x,y)$ if and only if
        \begin{align*}
            \frac{\partial f}{\partial z}(z) = 0 \quad \left( \text{or equivalently}\;\; \frac{\partial f}{\partial z^*}(z) = 0 \right).
        \end{align*}
    \end{enumerate}
\end{proposition}
\begin{proof}
    If $f(x,y)=u(x,y)+iv(x,y)$ is real-valued, $v(x,y)$ is necessarily the zero function. Then
    \begin{align*}
        \frac{\partial f}{\partial z/z^*}(z) = 0 &\iff \frac{1}{2}\left( \frac{\partial f}{\partial x}(x,y) \mp i \frac{\partial f}{\partial y}(x,y)\right) = 0\\
        &\iff \frac{\partial f}{\partial x}(x,y)=0 \;\;\text{and}\;\; \frac{\partial f}{\partial y}(x,y)=0.
    \end{align*}
    Furthermore, assume $f$ is non-constant. If $f$ is holomorphic, the Cauchy-Riemann equations $u_x=v_y=0$ and $u_y=-v_x=0$ imply $u(x,y)$ is constant, contradicting our assumption on $f$.
\end{proof}
As in Proposition \ref{Proposition:Wirtinger_derivatives_when_f_holomorphic}, we see the role of the Wirtinger derivatives as bookkeeping devices. Here $\partial f/\partial z = 0$ effectively merges the optimality conditions $\partial f/\partial x = 0$ and $\partial f/\partial y = 0$ into a single equation. The nature of the stationary point (i.e. whether it is a minimum/maximum/saddle point) has to be checked by inspecting higher-order derivatives or via additional considerations.

\begin{example}
    Let $f(z,z^*) = |z|^4-|z|^2$. We would like to find the stationary points of $f$. One could first convert it into $f(x,y) = (x^2+y^2)^2-(x^2+y^2)$. Setting the gradient of $f$ to zero gives
    \begin{align*}
        2x(2(x^2+y^2)-1) &= 0\\
        2y(2(x^2+y^2)-1) &= 0,
    \end{align*}
    which after some algebra gives $(x,y)=(0,0)$ and $\{(x,y): x^2+y^2=1/2\}$ as the stationary points.
    
    Using the Wirtinger derivatives, one avoids the hassle of converting $f(z,z^*)$ to $f(x,y)$, then carrying out differentiation twice. We simply obtain
    \begin{align*}
        \frac{\partial f}{\partial z}(z) = 0 &\implies z^*(2zz^*-1) = 0\\
        &\implies z=0 \text{  or  } |z|=1/\sqrt{2},
    \end{align*}
    clearly equivalent to what was previously obtained.
\end{example}

\subsection{Extension to multivariable complex functions}\label{Subsection:Extension_wirtinger_multivariable}
The extension of this formalism to functions $f: \mathbb{C}^n \longrightarrow \mathbb{C}$ is straightforward. We have
\begin{align*}
    f:\; &\mathbb{R}^{2n} \longrightarrow \mathbb{R}^2\\
    &(x_1,y_1,\dots,x_n,y_n) = (\mathbf{x},\mathbf{y}) \mapsto (u(\mathbf{x},\mathbf{y}),v(\mathbf{x},\mathbf{y})).
\end{align*}
Now for $i=1,\dots,n$ regard $z_i,z_i^*$ as functions from $\mathbb{R}^n \times \mathbb{R}^n$ to $\mathbb{C}$, where
$z_i(\mathbf{x},\mathbf{y}) = x_i+iy_i$ and $z_i^*(\mathbf{x},\mathbf{y}) = x_i-iy_i$. Then we have a function $\tilde{f}: \mathbb{C}^n \times \mathbb{C}^n \longrightarrow \mathbb{C}$ such that
\begin{align}\label{Equation:tildef_to_f_multivariable}
    f(\mathbf{x},\mathbf{y}) =: \tilde{f}\circ (\mathbf{z},\mathbf{z^*})(\mathbf{x},\mathbf{y}) = \tilde{f}(\mathbf{z}(\mathbf{x},\mathbf{y}),\mathbf{z^*}(\mathbf{x},\mathbf{y})) = \tilde{f}(\mathbf{x+iy},\mathbf{x-iy}).
\end{align}
Partial differentiating $f$ with respect to each $x_i$ and $y_i$ gives
\begin{align*}
    \frac{\partial f}{\partial x_i}(\mathbf{x},\mathbf{y}) &= \sum_{k=1}^n \frac{\partial \tilde{f}}{\partial z_k}\frac{\partial z_k}{\partial x_i} + \sum_{k=1}^n \frac{\partial \tilde{f}}{\partial z_k^*}\frac{\partial z_k^*}{\partial x_i} = \frac{\partial \tilde{f}}{\partial z_i}(\mathbf{z}(\mathbf{x},\mathbf{y}),\mathbf{z^*}(\mathbf{x},\mathbf{y})) + \frac{\partial \tilde{f}}{\partial z_i^*}(\mathbf{z}(\mathbf{x},\mathbf{y}),\mathbf{z^*}(\mathbf{x},\mathbf{y}))\\
    \frac{\partial f}{\partial y_i}(\mathbf{x},\mathbf{y}) &= \sum_{k=1}^n \frac{\partial \tilde{f}}{\partial z_k}\frac{\partial z_k}{\partial y_i} + \sum_{k=1}^n \frac{\partial \tilde{f}}{\partial z_k^*}\frac{\partial z_k^*}{\partial y_i} = i\cdot \frac{\partial \tilde{f}}{\partial z_i}(\mathbf{z}(\mathbf{x},\mathbf{y}),\mathbf{z^*}(\mathbf{x},\mathbf{y})) - i\cdot \frac{\partial \tilde{f}}{\partial z_i^*}(\mathbf{z}(\mathbf{x},\mathbf{y}),\mathbf{z^*}(\mathbf{x},\mathbf{y}))
\end{align*}
so after rearranging terms we obtain the Wirtinger derivatives
\begin{align*}
    \frac{\partial \tilde{f}}{\partial \mathbf{z}}(\mathbf{z}(\mathbf{x},\mathbf{y}),\mathbf{z^*}(\mathbf{x},\mathbf{y})) &= \frac{1}{2}\left(\frac{\partial f}{\partial \mathbf{x}} - i\frac{\partial f}{\partial \mathbf{y}}\right) (\mathbf{x},\mathbf{y})\\
    \frac{\partial \tilde{f}}{\partial \mathbf{z^*}}(\mathbf{z}(\mathbf{x},\mathbf{y}),\mathbf{z^*}(\mathbf{x},\mathbf{y})) &= \frac{1}{2}\left(\frac{\partial f}{\partial \mathbf{x}} + i\frac{\partial f}{\partial \mathbf{y}}\right) (\mathbf{x},\mathbf{y}).
\end{align*}
Here $\frac{\partial}{\partial \mathbf{z}} = \left( 
\frac{\partial}{\partial z_n}, \dots, \frac{\partial}{\partial z_1} \right)$ and similarly for $\frac{\partial}{\partial \mathbf{z^*}}$, $\frac{\partial}{\partial \mathbf{x}}$ and $\frac{\partial}{\partial \mathbf{y}}$. As before, we abuse notation and write both $f(\mathbf{x},\mathbf{y})$ and $f(\mathbf{z},\mathbf{z^*})$, so we can write
\begin{align}
    \frac{\partial}{\partial \mathbf{z}} = \frac{1}{2}\left(\frac{\partial}{\partial \mathbf{x}} - i\frac{\partial}{\partial \mathbf{y}}\right), \qquad \frac{\partial}{\partial \mathbf{z^*}} = \frac{1}{2}\left(\frac{\partial}{\partial \mathbf{x}} + i\frac{\partial}{\partial \mathbf{y}}\right).
\end{align}
The extension of Proposition \ref{Proposition:Optimization_with_Wirtinger_derivatives} also holds, namely if $f: \mathbb{C}^n \longrightarrow \mathbb{R}$ is real-valued then $f$ has a stationary point at $\mathbf{z}=(\mathbf{x},\mathbf{y})$ if and only if $\frac{\partial f}{\partial \mathbf{z/z^*}}(\mathbf{z}) = \mathbf{0}$.

\section{Matrix Wirtinger Calculus}
Having discussed Wirtinger derivatives for multivariable functions, we are now ready to examine the commonly encountered scenario where the parameters are structured as matrices.

\subsection{Matrix Wirtinger derivatives}
Now we consider functions of the form $f: \mathbb{C}^{n \times n} \longrightarrow \mathbb{C}$.\footnote{We could more generally let the domain be $\mathbb{C}^{m \times n}$. Indeed, many of the results developed below hold for such matrices as well. But for simplicity and because our matrices of interest in applications are square, we let $m=n$ here.} The following development, which we include for pedagogical purposes, is virtually the same as that in Section \ref{Subsection:Extension_wirtinger_multivariable}, just with different indexing. We have
\begin{align*}
    f:\; &\mathbb{R}^{2(n \times n)} \longrightarrow \mathbb{R}^2\\
    &(x_{ij},y_{ij})_{i,j \in [n]}=(\mathbf{X},\mathbf{Y}) \mapsto (u(\mathbf{X},\mathbf{Y}),v(\mathbf{X},\mathbf{Y})).
\end{align*} 
For $i=1,\dots,n$ regard $z_{ij},z_{ij}^*$ as functions from $\mathbb{R}^{n \times n} \times \mathbb{R}^{n \times n}$ to $\mathbb{C}$, where
$z_{ij}(\mathbf{X},\mathbf{Y}) = x_{ij}+iy_{ij}$ and $z_{ij}^*(\mathbf{X},\mathbf{Y}) = x_{ij}-iy_{ij}$. Then we have a function $\tilde{f}: \mathbb{C}^{n \times n} \times \mathbb{C}^{n \times n} \longrightarrow \mathbb{C}$ such that
\begin{align}\label{Equation:tildef_to_f_matrix}
    f(\mathbf{X},\mathbf{Y}) =: \tilde{f}\circ (\mathbf{Z},\mathbf{Z^*})(\mathbf{X},\mathbf{Y}) = \tilde{f}(\mathbf{Z}(\mathbf{X},\mathbf{Y}),\mathbf{Z^*}(\mathbf{X},\mathbf{Y})) = \tilde{f}(\mathbf{X+iY},\mathbf{X-iY}).
\end{align}
Again partial differentiating $f$ with respect to each $x_{ij}$ and $y_{ij}$ gives
\begin{align*}
    \frac{\partial f}{\partial x_{ij}}(\mathbf{X},\mathbf{Y}) &= \sum_{k,l=1}^n \frac{\partial \tilde{f}}{\partial z_{kl}}\frac{\partial z_{kl}}{\partial x_{ij}} + \sum_{k,l=1}^n \frac{\partial \tilde{f}}{\partial z_{kl}^*}\frac{\partial z_{kl}^*}{\partial x_{ij}}\\
    &= \sum_{k,l=1}^n \frac{\partial \tilde{f}}{\partial z_{kl}}\delta_{ki}\delta_{lj} + \sum_{k,l=1}^n \frac{\partial \tilde{f}}{\partial z_{kl}^*}\delta_{ki}\delta_{lj}\\
    &= \frac{\partial \tilde{f}}{\partial z_{ij}}(\mathbf{Z}(\mathbf{X},\mathbf{Y}),\mathbf{Z^*}(\mathbf{X},\mathbf{Y})) + \frac{\partial \tilde{f}}{\partial z_{ij}^*}(\mathbf{Z}(\mathbf{X},\mathbf{Y}),\mathbf{Z^*}(\mathbf{X},\mathbf{Y}))
\end{align*}
and
\begin{align*}
    \frac{\partial f}{\partial y_{ij}}(\mathbf{X},\mathbf{Y}) &= \sum_{k,l=1}^n \frac{\partial \tilde{f}}{\partial z_{kl}}\frac{\partial z_{kl}}{\partial y_{ij}} + \sum_{k,l=1}^n \frac{\partial \tilde{f}}{\partial z_{kl}^*}\frac{\partial z_{kl}^*}{\partial y_{ij}}\\ 
    &= i\cdot\frac{\partial \tilde{f}}{\partial z_{ij}}(\mathbf{Z}(\mathbf{X},\mathbf{Y}),\mathbf{Z^*}(\mathbf{X},\mathbf{Y})) - i\cdot\frac{\partial \tilde{f}}{\partial z_{ij}^*}(\mathbf{Z}(\mathbf{X},\mathbf{Y}),\mathbf{Z^*}(\mathbf{X},\mathbf{Y})).
\end{align*}
Rearranging terms as before gives us for $1 \leq i,j \leq n$
\begin{align}\label{Equation:after_rearranging}
    \frac{\partial \tilde{f}}{\partial z_{ij}}(\mathbf{Z}(\mathbf{X},\mathbf{Y}),\mathbf{Z^*}(\mathbf{X},\mathbf{Y})) &= \frac{1}{2}\left(\frac{\partial f}{\partial x_{ij}} - i\frac{\partial f}{\partial y_{ij}}\right) (\mathbf{X},\mathbf{Y})\\\nonumber
    \frac{\partial \tilde{f}}{\partial z_{ij}^*}(\mathbf{Z}(\mathbf{X},\mathbf{Y}),\mathbf{Z^*}(\mathbf{X},\mathbf{Y})) &= \frac{1}{2}\left(\frac{\partial f}{\partial x_{ij}} + i\frac{\partial f}{\partial y_{ij}}\right) (\mathbf{X},\mathbf{Y}).
\end{align}
To preserve the matrix structure of the parameters $z_{ij}$ and $z_{ij}^*$ we use the standard notation
\begin{align}\label{Equation:matrix_wirtinger_derivatives}
    \frac{\partial}{\partial \mathbf{Z}} := 
    \begin{bmatrix}
        \frac{\partial}{\partial z_{11}} & \dots & \frac{\partial}{\partial z_{1n}}\\
        \vdots & \ddots & \vdots\\
        \frac{\partial}{\partial z_{n1}} & \dots & \frac{\partial}{\partial z_{nn}}
    \end{bmatrix} \qquad
    \frac{\partial}{\partial \mathbf{Z^*}} := 
    \begin{bmatrix}
        \frac{\partial}{\partial z_{11}^*} & \dots & \frac{\partial}{\partial z_{1n}^*}\\
        \vdots & \ddots & \vdots\\
        \frac{\partial}{\partial z_{n1}^*} & \dots & \frac{\partial}{\partial z_{nn}^*}
    \end{bmatrix}
\end{align}
and similarly for $\frac{\partial}{\partial \mathbf{X}}$ and $\frac{\partial}{\partial \mathbf{Y}}$. Then Equation \ref{Equation:after_rearranging} is concisely stated as 
\begin{align}
    \frac{\partial \tilde{f}}{\partial \mathbf{Z}}(\mathbf{Z}(\mathbf{X},\mathbf{Y}),\mathbf{Z^*}(\mathbf{X},\mathbf{Y})) &= \frac{1}{2}\left(\frac{\partial f}{\partial \mathbf{X}} - i\frac{\partial f}{\partial \mathbf{Y}}\right) (\mathbf{X},\mathbf{Y})\\\nonumber
    \frac{\partial \tilde{f}}{\partial \mathbf{Z^*}}(\mathbf{Z}(\mathbf{X},\mathbf{Y}),\mathbf{Z^*}(\mathbf{X},\mathbf{Y})) &= \frac{1}{2}\left(\frac{\partial f}{\partial \mathbf{X}} + i\frac{\partial f}{\partial \mathbf{Y}}\right) (\mathbf{X},\mathbf{Y}).
\end{align}
$\frac{\partial}{\partial \mathbf{Z}}$ and $\frac{\partial}{\partial \mathbf{Z^*}}$ are the \textit{matrix Wirtinger derivatives} of $f$ and they are the protagonists of this article. As before, we abuse notation and write both $f(\mathbf{X},\mathbf{Y})$ and $f(\mathbf{Z},\mathbf{Z^*})$, so we can write
\begin{align}
    \frac{\partial}{\partial \mathbf{Z}} = \frac{1}{2}\left(\frac{\partial}{\partial \mathbf{X}} - i\frac{\partial}{\partial \mathbf{Y}}\right), \qquad \frac{\partial}{\partial \mathbf{Z^*}} = \frac{1}{2}\left(\frac{\partial}{\partial \mathbf{X}} + i\frac{\partial}{\partial \mathbf{Y}}\right).
\end{align}
The following matrix version of Proposition \ref{Proposition:Optimization_with_Wirtinger_derivatives} will be invoked frequently in optimization problems. The proof is omitted since it is a simple generalization of the proof in Proposition \ref{Proposition:Optimization_with_Wirtinger_derivatives}.
\begin{proposition}\label{Proposition:Optimization_with_Wirtinger_derivatives_matrix}
    Let $f: \mathbb{C}^{n \times n} \longrightarrow \mathbb{R}$ be a real-valued function of complex matrices. Then $f$ has a stationary point at $\mathbf{Z}=[z_{ij}]_{i,j \in [n]}$ if and only if
    \begin{align*}
        \frac{\partial f}{\partial \mathbf{Z}}(\mathbf{Z}) = 0 \quad \left( \text{or equivalently}\;\; \frac{\partial f}{\partial \mathbf{Z^*}}(\mathbf{Z}) = 0 \right).
    \end{align*}
\end{proposition}

\subsection{Functions of special interest}\label{Subsection:Functions_of_special_interest}
In this section we give several concrete examples of $f(\mathbf{Z},\mathbf{Z^*})$. Two such functions come to mind immediately---the trace and determinant of a matrix. We shall pay special attention to the trace function, as it appears most frequently in quantum information theory. For instance, the $l_1$-norm (trace), $l_2$-norm (Hilbert-Schmidt), Uhlmann fidelity, von Neumann entropy etc. are/can be expressed in terms of traces \cite{wilde2013quantum}. The determinant plays a less prominent role, so we focus less on that. Nonetheless, we provide a list of common $f(\mathbf{Z},\mathbf{Z^*})$ and their Wirtinger derivatives, see Table \ref{Table:List_Wirtinger_derivatives}.

\begin{table}
\centering
\setlength{\tabcolsep}{18pt}
\renewcommand{\arraystretch}{1.8}
\begin{tabular}{ |l l l| }
\hline
$f(\mathbf{Z,Z^*})$ & $\frac{\partial f}{\partial \mathbf{Z}}$ & $\frac{\partial f}{\partial \mathbf{Z^*}}$ \\ 
\hline
$\mathbf{a^TZb}$ & $\mathbf{ab}^T$ & $\mathbf{0}$ \\ 
$\mathbf{a^\dag Zb}$ & $\mathbf{a^*b}^T$ & $\mathbf{0}$ \\  
$\tr \mathbf{Z}$ & $\mathbf{I}$ & $\mathbf{0}$ \\ 
$\tr \mathbf{Z^*}$ & $\mathbf{0}$ & $\mathbf{I}$ \\  
$\tr(\mathbf{AZ})$ & $\mathbf{A}^T$ & $\mathbf{0}$ \\ 
$\tr(\mathbf{A}\mathbf{Z}^T)$ & $\mathbf{A}$ & $\mathbf{0}$ \\ 
$\tr(\mathbf{AZ^*})$ & $\mathbf{0}$ & $\mathbf{A}^T$ \\ 
$\tr(\mathbf{AZ^\dag})$ & $\mathbf{0}$ & $\mathbf{A}$ \\ 
$\tr(\mathbf{ZAZB})$ & $(\mathbf{AZB+BZA})^T$ & $\mathbf{0}$ \\ 
$\tr(\mathbf{ZAZ}^T\mathbf{B})$ & $\mathbf{B}^T\mathbf{ZA}^T+\mathbf{BZA}$ & $\mathbf{0}$ \\  
$\tr(\mathbf{ZAZ^*B})$ & $(\mathbf{AZ^*B})^T$ & $(\mathbf{BZA})^T$ \\ 
$\tr(\mathbf{ZAZ^\dag B})$ & $\mathbf{B}^T\mathbf{Z^*A}^T$ & $\mathbf{BZA}$ \\ 
$\tr(\mathbf{Z}^k)$ & $k(\mathbf{Z}^T)^{k-1}$ & $\mathbf{0}$ \\  
$\tr(\mathbf{AZ}^k)$ & $\sum_{r=0}^{k-1} (\mathbf{Z}^r\mathbf{AZ}^{k-1-r})^T$ & $\mathbf{0}$ \\ 
$\tr(F(\mathbf{Z}))$ & $F'(\mathbf{Z})^T$ & $\mathbf{0}$ \\ 
$\|\mathbf{Z}\|_F^2$ & $\mathbf{Z^*}$ & $\mathbf{Z}$ \\
$\det \mathbf{Z}$ & $\det(\mathbf{Z})(\mathbf{Z}^T)^{-1}$ & $\mathbf{0}$ \\ 
$\det \mathbf{Z^*}$ & $\mathbf{0}$ & $\det(\mathbf{Z^*})(\mathbf{Z}^\dag)^{-1}$ \\ 
$\det(\mathbf{Z^\dag AZ})$ & $\det(\mathbf{Z^\dag AZ})(\mathbf{Z}^T)^{-1}$ & $\det(\mathbf{Z^\dag AZ})(\mathbf{Z}^\dag)^{-1}$ \\  
$\det(\mathbf{Z}^k)$ & $k\det^k(\mathbf{Z})(\mathbf{Z}^T)^{-1}$ & $\mathbf{0}$ \\
\hline
\end{tabular}
\caption{A list of scalar functions $f(\mathbf{Z,Z^*})$ and their Wirtinger derivatives, assuming the input matrices $\mathbf{Z,Z^*}$ are unstructured (see next section). The items here are selected from Table 4.3, \cite{hjorungnes2011complex} and Sections 2 and 4, \cite{petersen2008matrix}.}
\label{Table:List_Wirtinger_derivatives}
\end{table}

From Section 2.5, \cite{petersen2008matrix}, we have that for real matrices $\mathbf{X}$ and a smooth function $F(x) = \sum_{n=0}^\infty c_nx^n$, the derivative of $\tr(F(\mathbf{X}))$ is given by
\begin{align*}
\frac{\partial \tr(F(\mathbf{X}))}{\partial \mathbf{X}} &= F'(\mathbf{X})^T.
\end{align*}
This is straightforwardly generalized to complex matrices $\mathbf{Z}$ and analytic functions $F(z) = \sum_{n=0}^\infty c_nz^n$. The following proposition will be heavily used in Section \ref{Subsection:Applications_in_QI} when we discuss applications in quantum information.
\begin{proposition}\label{Proposition:Wirtinger_derivative_of_traces}
Let $\mathbf{Z}$ be a complex matrix and $F(z) = \sum_{n=0}^\infty c_nz^n$ be analytic. Define the scalar function $f(\mathbf{Z,Z^*}):= \tr(F(\mathbf{Z}))$. Then
\begin{align*}
\frac{\partial \tr(F(\mathbf{Z}))}{\partial \mathbf{Z}} = F'(\mathbf{Z})^T
\end{align*}
where $F'(\cdot)$ is the complex derivative of $F(\cdot)$.
\end{proposition}
\begin{proof}
The operations of summing, tracing and differentiating commute, so we have
\begin{alignat*}{5}
\frac{\partial \tr(F(\mathbf{Z}))}{\partial Z_{ij}} &= \sum_{n=0}^\infty c_n \tr \frac{\partial \mathbf{Z}^n}{\partial Z_{ij}}\\
&= \sum_{n=1}^\infty c_n \sum_l \sum_{r=0}^{n-1} \left(\mathbf{Z}^r \Delta_{ij} \mathbf{Z}^{n-1-r} \right)_{ll} \quad &&(\text{see lemma below})\\
&= \sum_{n=1}^\infty c_n \sum_{r=0}^{n-1} \sum_l \left(\Delta_{ij} \mathbf{Z}^{n-1} \right)_{ll} \quad &&(\text{cyclic permutation within trace})\\
&= \sum_{n=1}^\infty c_n \sum_{r=0}^{n-1} \sum_l \sum_p \Delta_{ij,lp} (\mathbf{Z}^{n-1})_{pl}\\
&= \sum_{n=1}^\infty c_n \sum_{r=0}^{n-1} \sum_l \sum_p \delta_{il}\delta_{jp} (\mathbf{Z}^{n-1})_{pl} \quad &&(\text{definition of $\Delta_{ij}$})\\
&= \sum_{n=1}^\infty c_n \sum_{r=0}^{n-1} (\mathbf{Z}^{n-1})_{ji}\\
&= \sum_{n=1}^\infty c_n n(\mathbf{Z}^{n-1})_{ji}\\
&= F'(\mathbf{Z})_{ji}.
\end{alignat*}
\end{proof}

\begin{lemma}[pg11, \cite{petersen2008matrix}]
For $n \geq 1$,
\begin{align*}
\frac{\partial (\mathbf{Z}^n)_{lm}}{\partial Z_{ij}} = \sum_{r=0}^{n-1} \left(\mathbf{Z}^r \Delta_{ij} \mathbf{Z}^{n-1-r} \right)_{lm}
\end{align*}
where $\Delta_{ij}$ is the matrix with entry $(i,j)$ equaling one and other entries zero, i.e. $\Delta_{ij,lp} := \delta_{il}\delta_{jp}$.
\end{lemma}
\begin{proof}
Expanding the matrix product gives
\begin{align*}
\frac{\partial (\mathbf{Z}^n)_{lm}}{\partial Z_{ij}} &= \frac{\partial}{\partial Z_{ij}} \left( \sum_{k_1,...,k_{n-1}} Z_{lk_1}Z_{k_1k_2}...Z_{k_{n-2}k_{n-1}} Z_{k_{n-1}m} \right)\\
&= \sum_{k_1,...,k_{n-1}} \sum_{r=0}^{n-1} Z_{lk_1}Z_{k_1k_2}...Z_{k_{r-1}k_r}\delta_{ik_r}\delta_{jk_{r+1}}Z_{k_{r+1}k_{r+2}}...Z_{k_{n-1}m} \quad (\text{product rule})\\
&= \sum_{r=0}^{n-1} \sum_{k_1,...,k_{n-1}} Z_{lk_1}Z_{k_1k_2}...Z_{k_{r-1}k_r}\Delta_{ij,k_rk_{r+1}}Z_{k_{r+1}k_{r+2}}...Z_{k_{n-1}m}\\
&= \sum_{r=0}^{n-1} \left(\mathbf{Z}^r \Delta_{ij} \mathbf{Z}^{n-1-r} \right)_{lm}
\end{align*}
where $k_0=l$ and $k_n=m$ in the second line.
\end{proof}

\subsection{Structured matrices and the chain rule}
So far, by writing $f: \mathbb{C}^{n \times n} \longrightarrow \mathbb{C}$ we have implicitly assumed the input matrices have independent components. This condition often does not hold, e.g. when our matrices of interest are symmetric/Hermitian etc. Such matrices are termed `patterned' in \cite{hjorungnes2011complex} and `structured' in \cite{petersen2008matrix}) (we shall adopt the latter terminology). 

In this situation, the results in Table \ref{Table:List_Wirtinger_derivatives} no longer hold. To obtain the correct Wirtinger derivatives with respect to structured matrices, we resort to the chain rule. For pedagogical purposes, we first demonstrate how to use the chain rule to find the matrix derivative of a function with respect to real symmetric matrices. Then we do likewise to obtain the Wirtinger derivatives of a function of complex, Hermitian matrices. In Table \ref{Table:Structured_Wirtinger_derivatives} we give the Wirtinger derivatives of $f(\mathbf{Z,Z^*})$ when $\mathbf{Z}$ is diagonal, symmetric, anti-symmetric, Hermitian and anti-Hermitian. All can be derived from the chain rule.

Let $f(\mathbf{X})$ be a function of real matrices. The matrix derivative of $f(\mathbf{X})$ with respect to $\mathbf{X}$ has components $\partial f/\partial X_{ij}$. Using the chain rule, we write
\begin{align}\label{Equation:Structure_matrix_chain_rule}
    \frac{\partial f}{\partial X_{ij}} = \sum_{k,l=1}^n \frac{\partial f}{\partial \tilde{X}_{kl}} \frac{\partial \tilde{X}_{kl}}{\partial X_{ij}}.
\end{align}
Here, the tilde above $\mathbf{\tilde{X}}$ indicates that it is an \textit{unstructured} matrix. The term $\partial \tilde{X}_{kl}/\partial X_{ij}$ takes into account the dependencies of the components of $\mathbf{X}$. Following the terminology in \cite{petersen2008matrix}, we call it the \textit{structure matrix} of $\mathbf{X}$. If $\mathbf{X}$ is unstructured, then its structure matrix is simply
\begin{align*}
    \frac{\partial \tilde{X}_{kl}}{\partial X_{ij}} = \delta_{ki}\delta_{lj}.
\end{align*}

Now let $\mathbf{X}$ be symmetric. For diagonal matrix elements we have $\partial \tilde{X}_{kl}/\partial X_{ij} = \delta_{kl}\delta_{ij}$, while for off-diagonal matrix elements we have $\partial \tilde{X}_{kl}/\partial X_{ij} = \delta_{kl}\delta_{ij} + \delta_{li}\delta_{kj}$. We can write these succinctly as
\begin{align*}
    \frac{\partial \tilde{X}_{kl}}{\partial X_{ij}} &= \delta_{ki}\delta_{lj} + (1-\delta_{kl})\delta_{li}\delta_{kj}\\
    &= \delta_{ki}\delta_{lj} + \delta_{li}\delta_{kj} - \delta_{kl}\delta_{li}\delta_{kj}.
\end{align*}
Plugging this into Eq. \ref{Equation:Structure_matrix_chain_rule}, we obtain
\begin{align*}
    \frac{\partial f}{\partial X_{ij}} &= \sum_{k,l=1}^n \frac{\partial f}{\partial \tilde{X}_{kl}} \frac{\partial \tilde{X}_{kl}}{\partial X_{ij}}\\
    &= \sum_{k,l=1}^n \frac{\partial f}{\partial \tilde{X}_{kl}} (\delta_{ki}\delta_{lj} + \delta_{li}\delta_{kj} - \delta_{kl}\delta_{li}\delta_{kj})\\
    &= \frac{\partial f}{\partial \tilde{X}_{ij}} + \frac{\partial f}{\partial \tilde{X}_{ji}} - \delta_{ij}\frac{\partial f}{\partial \tilde{X}_{ij}}
\end{align*}
for $1\leq i,j \leq n$. This gives
\begin{proposition}[Matrix derivative with respect to symmetric matrices]\label{Proposition:Wirtinger_symmetric}
    Let $f(\mathbf{X})$ be a function of real symmetric matrices. Then the matrix derivative of $f$ with respect to $\mathbf{X}$ is
    \begin{align*}
        \frac{\partial f}{\partial \mathbf{X}} = \left[\frac{\partial f}{\partial \mathbf{\tilde{X}}} + \left(\frac{\partial f}{\partial \mathbf{\tilde{X}}}\right)^T - \mathbf{I} \odot \frac{\partial f}{\partial \mathbf{\tilde{X}}}\right]_{\mathbf{\tilde{X}}=\mathbf{X}}.
    \end{align*}
\end{proposition}
Note that $\frac{\partial f}{\partial X_{ij}} = \frac{\partial f}{\partial X_{ji}}$, as expected. Thus, to derive the matrix derivative with respect to symmetric matrices, we first obtain the matrix derivative of $f$, \textit{assuming the inputs are unstructured}, then form the linear combination given in Proposition \ref{Proposition:Wirtinger_symmetric}, before reinstating the structured matrix $\mathbf{X}$ as the argument. For example,
\begin{align*}
    \frac{\partial \tr(\mathbf{AX})}{\partial \mathbf{X}} = \mathbf{A} + \mathbf{A}^T - \mathbf{I \odot A}
\end{align*}
and
\begin{align*}
    \frac{\partial \tr(\mathbf{X}^2)}{\partial \mathbf{X}} = 2\mathbf{X} + 2\mathbf{X}^T - 2\mathbf{I \odot \mathbf{X}^T}.
\end{align*}

\begin{table}[t]
\centering
\setlength{\tabcolsep}{10pt}
\renewcommand{\arraystretch}{2.0}
\begin{tabular}{ |l l l| }
\hline
Structure of $\mathbf{Z}$ & $\frac{\partial f}{\partial \mathbf{Z}}(\mathbf{Z,Z^*})$ & $\frac{\partial f}{\partial \mathbf{Z}^*}(\mathbf{Z,Z^*})$\\ 
\hline
Unstructured: $\mathbf{Z} \in \mathbb{C}^{n \times n}$ & $\left[\frac{\partial f}{\partial \mathbf{\tilde{Z}}}\right]_{\mathbf{\tilde{Z}}=\mathbf{Z}}$ & $\left[\frac{\partial f}{\partial \mathbf{\tilde{Z}^*}}\right]_{\mathbf{\tilde{Z}}=\mathbf{Z}}$\\ 
Diagonal: $\mathbf{Z} = \mathbf{I} \odot \mathbf{Z}$ & $\left[\mathbf{I} \odot \frac{\partial f}{\partial \mathbf{\tilde{Z}}}\right]_{\mathbf{\tilde{Z}}=\mathbf{Z}}$ & $\left[\mathbf{I} \odot \frac{\partial f}{\partial \mathbf{\tilde{Z}^*}}\right]_{\mathbf{\tilde{Z}}=\mathbf{Z}}$\\ 
Symmetric: $\mathbf{Z} = \mathbf{Z}^T$ & $\left[\frac{\partial f}{\partial \mathbf{\tilde{Z}}} + \left(\frac{\partial f}{\partial \mathbf{\tilde{Z}}}\right)^T - \mathbf{I} \odot \frac{\partial f}{\partial \mathbf{\tilde{Z}}}\right]_{\mathbf{\tilde{Z}}=\mathbf{Z}}$ & $\left[\frac{\partial f}{\partial \mathbf{\tilde{Z}^*}} + \left(\frac{\partial f}{\partial \mathbf{\tilde{Z}^*}}\right)^T - \mathbf{I} \odot \frac{\partial f}{\partial \mathbf{\tilde{Z}^*}}\right]_{\mathbf{\tilde{Z}}=\mathbf{Z}}$\\ 
Anti-symmetric: $\mathbf{Z} = -\mathbf{Z}^T$ & $\left[\frac{\partial f}{\partial \mathbf{\tilde{Z}}} - \left(\frac{\partial f}{\partial \mathbf{\tilde{Z}}}\right)^T\right]_{\mathbf{\tilde{Z}}=\mathbf{Z}}$ & $\left[\frac{\partial f}{\partial \mathbf{\tilde{Z}^*}} - \left(\frac{\partial f}{\partial \mathbf{\tilde{Z}^*}}\right)^T\right]_{\mathbf{\tilde{Z}}=\mathbf{Z}}$\\ 
Hermitian: $\mathbf{Z} = \mathbf{Z}^\dag$ & $\left[\frac{\partial f}{\partial \mathbf{\tilde{Z}}} + \left(\frac{\partial f}{\partial \mathbf{\tilde{Z}^*}}\right)^T\right]_{\mathbf{\tilde{Z}}=\mathbf{Z}}$ & $\left[\frac{\partial f}{\partial \mathbf{\tilde{Z}^*}} + \left(\frac{\partial f}{\partial \mathbf{\tilde{Z}}}\right)^T\right]_{\mathbf{\tilde{Z}}=\mathbf{Z}}$\\ 
Anti-Hermitian: $\mathbf{Z} = -\mathbf{Z}^\dag$ & $\left[\frac{\partial f}{\partial \mathbf{\tilde{Z}}} - \left(\frac{\partial f}{\partial \mathbf{\tilde{Z}^*}}\right)^T\right]_{\mathbf{\tilde{Z}}=\mathbf{Z}}$ & $\left[\frac{\partial f}{\partial \mathbf{\tilde{Z}^*}} - \left(\frac{\partial f}{\partial \mathbf{\tilde{Z}}}\right)^T\right]_{\mathbf{\tilde{Z}}=\mathbf{Z}}$\\
\hline
\end{tabular}
\caption{The Wirtinger derivatives of $f(\mathbf{Z,Z^*})$ where $\mathbf{Z}$ is a structured matrix. The tildes above $\mathbf{\tilde{Z},\tilde{Z}^*}$ indicate that they are unstructured matrices. To derive the Wirtinger derivatives with respect to structured matrices, first obtain the Wirtinger derivative of $f$, assuming the inputs are unstructured. After forming the correct expressions given above, reinstate the structured matrices $\mathbf{Z,Z^*}$ as the arguments. Here we consider five structure classes, namely the diagonal, symmetric, anti-symmetric, Hermitian and anti-Hermitian matrices. This table is replicated from Table 6.2, \cite{hjorungnes2011complex}.}
\label{Table:Structured_Wirtinger_derivatives}
\end{table}

Having demonstrated how to find matrix derivatives with respect to real structured matrices, we now apply the same methodology to obtain the Wirtinger derivatives of $f(\mathbf{Z,Z^*})$ when $\mathbf{Z}$ is structured. Let $f(\mathbf{Z,Z^*})$ be a function of complex matrices. Applying the chain rule, the Wirtinger derivatives are given by
\begin{align}\label{Equation:Structure_matrix_chain_rule_Wirtinger}
    \frac{\partial f}{\partial Z_{ij}} &= \sum_{k,l=1}^n \frac{\partial f}{\partial \tilde{Z}_{kl}} \frac{\partial \tilde{Z}_{kl}}{\partial Z_{ij}} + \sum_{k,l=1}^n \frac{\partial f}{\partial \tilde{Z}_{kl}^*} \frac{\partial \tilde{Z}_{kl}^*}{\partial Z_{ij}}\\\nonumber
    \frac{\partial f}{\partial Z_{ij}^*} &= \sum_{k,l=1}^n \frac{\partial f}{\partial \tilde{Z}_{kl}} \frac{\partial \tilde{Z}_{kl}}{\partial Z_{ij}^*} + \sum_{k,l=1}^n \frac{\partial f}{\partial \tilde{Z}_{kl}^*} \frac{\partial \tilde{Z}_{kl}^*}{\partial Z_{ij}^*}.
\end{align}
As before, the tildes above $\mathbf{\tilde{Z},\tilde{Z}^*}$ indicate that they are unstructured matrices, so $\frac{\partial f}{\partial \mathbf{\tilde{Z}}}$, $\frac{\partial f}{\partial \mathbf{\tilde{Z}^*}}$ can be obtained using the results in Section \ref{Subsection:Functions_of_special_interest}. Note that there are now four structure matrices. If $\mathbf{Z}$ is unstructured, then they are simply
\begin{align*}
    \frac{\partial \tilde{Z}_{kl}}{\partial Z_{ij}} = \delta_{ki}\delta_{lj} \qquad \frac{\partial \tilde{Z}_{kl}}{\partial Z_{ij}^*} = 0 \qquad \frac{\partial \tilde{Z}_{kl}^*}{\partial Z_{ij}} = 0 \qquad \frac{\partial \tilde{Z}_{kl}^*}{\partial Z_{ij}^*} = \delta_{ki}\delta_{lj}.
\end{align*}
If $\mathbf{Z}$ is Hermitian, we have
\begin{proposition}[Wirtinger derivatives with respect to Hermitian matrices]\label{Proposition:Wirtinger_hermitian}
    Let $f(\mathbf{Z,Z^*})$ be a function of complex Hermitian matrices. Then the Wirtinger derivatives of $f$ with respect to $\mathbf{Z,Z^*}$ are given by
    \begin{align*}
        \frac{\partial f}{\partial \mathbf{Z}} = \left[\frac{\partial f}{\partial \mathbf{\tilde{Z}}} + \left(\frac{\partial f}{\partial \mathbf{\tilde{Z}^*}}\right)^T\right]_{\mathbf{\tilde{Z}}=\mathbf{Z}} \qquad \text{and} \qquad
        \frac{\partial f}{\partial \mathbf{Z^*}} = \left[\frac{\partial f}{\partial \mathbf{\tilde{Z}^*}} + \left(\frac{\partial f}{\partial \mathbf{\tilde{Z}}}\right)^T\right]_{\mathbf{\tilde{Z}}=\mathbf{Z}}.
    \end{align*}
\end{proposition}
\begin{proof}
    Let us first consider the Wirtinger derivative $\frac{\partial f}{\partial \mathbf{Z}}$. We have the structure matrices
    \begin{align*}
        \frac{\partial \tilde{Z}_{kl}}{\partial Z_{ij}} = \delta_{ki}\delta_{lj} \qquad \text{and} \qquad \frac{\partial \tilde{Z}_{kl}^*}{\partial Z_{ij}} = \delta_{li}\delta_{kj}.
    \end{align*}
    Thus
    \begin{align*}
        \frac{\partial f}{\partial Z_{ij}} &= \sum_{k,l=1}^n \frac{\partial f}{\partial \tilde{Z}_{kl}} \frac{\partial \tilde{Z}_{kl}}{\partial Z_{ij}} + \sum_{k,l=1}^n \frac{\partial f}{\partial \tilde{Z}_{kl}^*} \frac{\partial \tilde{Z}_{kl}^*}{\partial Z_{ij}}\\
        &= \sum_{k,l=1}^n \frac{\partial f}{\partial \tilde{Z}_{kl}} \delta_{ki}\delta_{lj} + \sum_{k,l=1}^n \frac{\partial f}{\partial \tilde{Z}_{kl}^*} \delta_{li}\delta_{kj}\\
        &= \frac{\partial f}{\partial \tilde{Z}_{ij}} + \frac{\partial f}{\partial \tilde{Z}_{ji}^*}.
    \end{align*}
    $\frac{\partial f}{\partial \mathbf{Z^*}}$ can be similarly obtained. Alternatively, because $\mathbf{Z^*} = \mathbf{Z}^T$ it follows that
    \begin{align*}
        \frac{\partial f}{\partial \mathbf{Z^*}} = \frac{\partial f}{\partial \mathbf{Z}^T} = \left( \frac{\partial f}{\partial \mathbf{Z}} \right)^T = \left[\frac{\partial f}{\partial \mathbf{\tilde{Z}}} + \left(\frac{\partial f}{\partial \mathbf{\tilde{Z}^*}}\right)^T\right]_{\mathbf{\tilde{Z}}=\mathbf{Z}}^T = \left[\frac{\partial f}{\partial \mathbf{\tilde{Z}^*}} + \left(\frac{\partial f}{\partial \mathbf{\tilde{Z}}}\right)^T\right]_{\mathbf{\tilde{Z}}=\mathbf{Z}}.
    \end{align*}
    Note that $\frac{\partial f}{\partial Z_{ij}} = \frac{\partial f}{\partial Z_{ji}^*}$, as expected.
\end{proof}

\section{Applications in Quantum Information}\label{Subsection:Applications_in_QI}
To illustrate the utility of this calculus, we discuss a few examples of relevance in quantum information.

\subsection{Examples}
\begin{example}\label{Example:purity}
Consider the purity function on density operators, $f: \mathcal{D}(\mathcal{H}_N) \longrightarrow \mathbb{R}$, $\rho \mapsto \tr(\rho^2)$. The purity of a density operator/state quantifies the noisiness of a state \cite{wilde2013quantum}. Which state $\rho$ has the lowest purity? Intuitively, it is the maximally mixed state. Let us verify this via two approaches: the `classical' method and Wirtinger's method. 

Since this is our first example, we shall be more thorough with the details.
\begin{enumerate}[i.]
\item (Classical) For any $\rho \in \mathcal{D}(\mathcal{H}_N)$ we can write $\rho_{ii}=a_{ii}$ and $\rho_{ij}=a_{ij}+ib_{ij}=\rho^*_{ji}$ for $i \neq j$, where $a_{ij},b_{ij} \in \mathbb{R}$. For example if $N=2$ then
\begin{align*}
    \rho = 
    \begin{bmatrix}
    \rho_{11} & \rho_{12}\\
    \rho_{21} & \rho_{22}
    \end{bmatrix}
    =
    \begin{bmatrix}
    a_{11} & a_{12}+ib_{12}\\
    a_{12}-ib_{12} & a_{22}
    \end{bmatrix}.
\end{align*}
Thus
\begin{align*}
    \tr(\rho^2) = \sum_{i,j=1}^N \braket{i|\rho|j}\braket{j|\rho|i} &= \sum_{i,j=1}^N |\braket{i|\rho|j}|^2\\
    &= \sum_{i=1}^N |\rho_{ii}|^2 + \sum_{i \neq j} |\rho_{ij}|^2\\
    &= \sum_{i=1}^N a_{ii}^2 + 2\sum_{i<j} (a_{ij}^2+b_{ij}^2).
\end{align*}
The purity function is viewed as a function of $N^2$ real inputs:
\begin{align*}
    f(\{a_{ii}\}_{i \in [N]}; \{a_{ij},b_{ij}\}_{i<j}) = \sum_{i=1}^N a_{ii}^2 + 2\sum_{i<j} (a_{ij}^2+b_{ij}^2).
\end{align*}
We formulate the Lagrangian function
\begin{align*}
    \mathcal{L} = \sum_i a_{ii}^2 + 2\sum_{i<j} (a_{ij}^2+b_{ij}^2) - \lambda \left( \sum_i a_{ii} - 1 \right),
\end{align*}
where $\lambda$ is the Lagrange multiplier taking into account the constraint $\tr \rho = 1$. The first-order conditions give
\begin{align*}
    \frac{\partial \mathcal{L}}{\partial a_{ii}} = 0 &\implies a_{ii}=\frac{\lambda}{2}\\
    \frac{\partial \mathcal{L}}{\partial a_{ij}} = 0 &\implies a_{ij}=0\\
    \frac{\partial \mathcal{L}}{\partial \lambda} = 0 &\implies \lambda = \frac{2}{N},
\end{align*}
therefore the optimal state is $\rho = \frac{I}{N}$, the maximally mixed state.

\item (Wirtinger) Per Wirtinger's approach, we retain the structure of $\rho$. Here $f(\rho,\rho^*) = \tr(\rho^2)$ where $\rho$ is Hermitian and thus structured. We formulate the Lagrangian: 
\begin{align*}
    \mathcal{L} = \tr(\rho^2) - \lambda(\tr \rho-1).
\end{align*}
To find the optimal $\rho$, we invoke Proposition \ref{Proposition:Optimization_with_Wirtinger_derivatives_matrix}. The conditions required are
\begin{align*}
    \frac{\partial \mathcal{L}}{\partial \rho} &= \frac{\partial \tr(\rho^2)}{\partial \rho} - \lambda \frac{\partial (\tr \rho-1)}{\partial \rho} = \mathbf{0}\\
    \frac{\partial \mathcal{L}}{\partial \lambda} &= 0.
\end{align*}
To obtain $\frac{\partial \tr(\rho^2)}{\partial \rho}$ and $\frac{\partial \tr \rho}{\partial \rho}$ we use
Proposition \ref{Proposition:Wirtinger_derivative_of_traces} (or simply consult Table \ref{Table:List_Wirtinger_derivatives}) together with the chain rule, see Proposition \ref{Proposition:Wirtinger_symmetric}. For example
\begin{align*}
    \frac{\partial \tr(\rho^2)}{\partial \rho} &= \left[\frac{\partial \tr(\tilde{\rho}^2)}{\partial \tilde{\rho}} + \cancel{ \left(\frac{\partial \tr(\tilde{\rho}^2)}{\partial \tilde{\rho}^*}\right)^T}\right]_{\tilde{\rho}=\rho}\\
    &= 2\rho^T,
\end{align*}
and similarly $\frac{\partial \tr \rho}{\partial \rho} = I$. Therefore $\frac{\partial \mathcal{L}}{\partial \rho} = \mathbf{0} \implies 2\rho^T-\lambda I = \mathbf{0} \implies \rho = \frac{\lambda}{2}I$. The constraint $\tr \rho = 1$ then gives $\rho = \frac{I}{N}$.
\end{enumerate}
\end{example}

\begin{example}\label{Example:vNEntropy}
Which density operator $\rho$ maximizes the von Neumann entropy $H(\rho)=-\tr(\rho \log \rho)$? The usual textbook solution makes use of the quantum data-processing inequality, which states that $S(\sigma\|\rho) \geq S(\mathcal{N}(\sigma)\|\mathcal{N}(\rho))$ for quantum states $\sigma,\rho$ and any quantum channel $\mathcal{N}$ (here $S(\sigma\|\rho)$ is the quantum relative entropy). Defining the quantum channel $\mathcal{E} : \mathcal{D}(\mathcal{H}_d) \longrightarrow \mathcal{D}(\mathcal{H}_d)$, $\rho \mapsto \frac{I}{d}$, we have 
$
-H(\rho)+\log d = D(\rho\|\frac{I}{d}) \geq D(\mathcal{E}(\rho)\|\mathcal{E}(\frac{I}{d})) = D(\frac{I}{d}\|\frac{I}{d}) = 0.
$
Thus $H(\rho) \leq \log d$, which is attained when $\rho=\frac{I}{d}$.

Using Wirtinger calculus, we first prepare the Lagrangian $\mathcal{L} = -\tr(\rho \log \rho) - \lambda(\tr(\rho)-1)$. $\frac{\partial \mathcal{L}}{\partial \rho} = \mathbf{0} \implies 1+\log \rho^T + \lambda = \mathbf{0} \implies \rho=e^{-1-\lambda}I \implies \rho = \frac{I}{d}$.
\end{example}

\begin{example}\label{Example:constrained_vNEntropy}
We extend the example above by imposing an additional constraint on $\rho$, namely that $\tr(\rho H)=E$ for the Hamiltonian $H$ governing the quantum system. So we want to solve
\bex
\text{minimize}_{\rho \in \mathcal{D}(\mathcal{H})} && H(\rho) \\\nonumber
\text{s.t.} && \tr (\rho H) = E.
\eex
Note that because $\rho,H$ are Hermitian, $\tr(\rho H)$ is real. We require $h_{\min} < E < h_{\max}$ ($h$ denotes a generic eigenvalue of $H$), otherwise the constraint $\tr(\rho H)=E$ cannot be satisfied.

Set up the Lagrangian
\begin{align*}
    \mathcal{L} = -\tr(\rho \log \rho) -  \beta (\tr(\rho H)-E) - \eta(\tr\rho-1),
\end{align*}
where $\beta$ and $\eta$ are the Lagrange multipliers. Making use of Proposition \ref{Proposition:Wirtinger_derivative_of_traces} for trace functions and the chain rule, Proposition \ref{Proposition:Wirtinger_hermitian}, setting the Wirtinger derivative of $\mathcal{L}$ to zero gives
\begin{align*}
\frac{\partial \mathcal{L}}{\partial \rho} = \mathbf{0} &\implies -(\log \rho)^T - I - \beta H^T - \eta I = \mathbf{0}\\
&\implies \rho = e^{-(\eta+1)}e^{-\beta H}\\
&\implies \rho = \frac{e^{-\beta H}}{\tr(e^{-\beta H})} \qquad \text{after normalization},
\end{align*}
where $\beta$ is to be determined from the constraint $\tr(\rho H) = E$. We see that the solution takes the form of a Gibbs state with respect to the Hamiltonian $H$ and inverse temperature $\beta$ \cite{pathria2021statistical,kardar2007statistical}.
\end{example}

\begin{example}\label{Example:Frobenius}
In this example, we consider the Frobenius/Hilbert-Schmidt norm $\|\cdot\|_F$ as the figure of merit. Given a linear operator $L \in \mathcal{L}(\mathcal{H})$, we want to find the operator $T$ such that the Frobenius metric $\|T-L\|_F$ is minimized, while keeping the Frobenius norm and trace of $T$ fixed. That is, we want to solve
\bex
    \text{minimize}_{T \in \mathcal{L}(\mathcal{H})} && \|T-L\|_F \\\nonumber
    \text{s.t.} && \|T\|_F=C \\\nonumber
    && \tr T = D.
\eex
\end{example}
The Frobenius norm can be expressed as a trace: $\|T\|_F := \sqrt{\tr(T^\dag T)}$. Since the map $x \mapsto x^2$ is increasing on the positive numbers, we could equally well minimize $\|T\|_F^2$---this removes the square root. The Lagrangian is then
\begin{align*}
    \mathcal{L} &= \|T-L\|_F^2 + \lambda(\|T\|_F^2-C) + \eta(\tr T - D)\\
    &= \tr\left( (T-L)^\dag(T-L) \right) + \lambda(\tr(T^\dag T)-C^2) + \eta(\tr T - D).
\end{align*}
Then, we obtain
\begin{align*}
    \frac{\partial \mathcal{L}}{\partial T} = \mathbf{0} &\implies T^*-L^*+\lambda T^*+\eta I = \mathbf{0} \qquad \text{(cf.~Table \ref{Table:List_Wirtinger_derivatives})}\\
    &\implies T = \frac{L-\eta I}{1+\lambda}.
\end{align*}
The Lagrange multipliers $\lambda$ and $\eta$ can then be obtained from the constraint equations. We note that if there were no trace constraint then $T = \frac{C}{\|L\|_F}L$, while if there were no Frobenius norm constraint then $T = L - \frac{\tr L - D}{\tr I}I$.

\subsection{Discussion}
Above we gave a few simple applications of Wirtinger Calculus in quantum information. For more involved applications of this framework in quantum information, we direct the reader to the recent papers on quantum tomography \cite{utreras2019stochastic} and quantum estimation theory \cite{munoz2022complex}. 

The WC first found widespread practical use in electrical engineering. The interested reader can consult \cite{hjorungnes2011complex} and \cite{kreutz2009complex} for the myriad applications of WC in various topics, such as digital signal processing, wireless communications, control theory and adaptive filtering. More recently the WC is also being widely utilized in machine learning (broadly defined). See for example \cite{bahmani2017phase,bouboulis2010extension,li2010algorithms,pu2021input,virtue2019complex,zhang2019fully}. We hope that our little article would be of use to researchers in these communities as well.

\section*{Acknowledgments}
This work is supported
by the National Research Foundation, Singapore, and A*STAR under its CQT Bridging Grant and its Quantum
Engineering Programme under grant NRF2021-QEP2-02-P05.

%%%%%%%%%%%%%%%%%%%%%%%%%%%%%%%%%%%%%%%%%%%%%
\addcontentsline{toc}{section}{References}
\bibliographystyle{alpha}
\bibliography{main}

\newcommand{\etalchar}[1]{$^{#1}$}
\begin{thebibliography}{UARTND19}

\bibitem[BR17]{bahmani2017phase}
Sohail Bahmani and Justin Romberg.
\newblock Phase retrieval meets statistical learning theory: A flexible convex relaxation.
\newblock In {\em Artificial Intelligence and Statistics}, pages 252--260. PMLR, 2017.

\bibitem[Bra83]{brandwood1983complex}
DH~Brandwood.
\newblock A complex gradient operator and its application in adaptive array theory.
\newblock In {\em IEE Proceedings H (Microwaves, Optics and Antennas)}, volume 130, pages 11--16. IET Digital Library, 1983.

\bibitem[BT10]{bouboulis2010extension}
Pantelis Bouboulis and Sergios Theodoridis.
\newblock {Extension of Wirtinger's calculus to reproducing kernel Hilbert spaces and the complex kernel LMS}.
\newblock {\em IEEE Transactions on Signal Processing}, 59(3):964--978, 2010.

\bibitem[GK06]{greene2006function}
Robert~Everist Greene and Steven~George Krantz.
\newblock {\em Function theory of one complex variable}.
\newblock American Mathematical Soc., 2006.

\bibitem[HG07]{hjorungnes2007complex}
Are Hjorungnes and David Gesbert.
\newblock Complex-valued matrix differentiation: Techniques and key results.
\newblock {\em IEEE Transactions on Signal Processing}, 55(6):2740--2746, 2007.

\bibitem[Hj{\o}11]{hjorungnes2011complex}
Are Hj{\o}rungnes.
\newblock {\em Complex-valued matrix derivatives: with applications in signal processing and communications}.
\newblock Cambridge University Press, 2011.

\bibitem[Kar07]{kardar2007statistical}
Mehran Kardar.
\newblock {\em Statistical physics of particles}.
\newblock Cambridge University Press, 2007.

\bibitem[KD09]{kreutz2009complex}
Ken Kreutz-Delgado.
\newblock The complex gradient operator and the {CR}-calculus.
\newblock {\em arXiv preprint arXiv:0906.4835}, 2009.

\bibitem[LA10]{li2010algorithms}
Hualiang Li and T{\"u}lay~T{\"u}lay Adali.
\newblock Algorithms for complex {ML ICA} and their stability analysis using {W}irtinger calculus.
\newblock {\em IEEE Transactions on Signal Processing}, 58(12):6156--6167, 2010.

\bibitem[MPND22]{munoz2022complex}
M~Mu{\~n}oz, L~Pereira, S~Niklitschek, and A~Delgado.
\newblock Complex field formulation of the quantum estimation theory.
\newblock {\em arXiv preprint arXiv:2203.03064}, 2022.

\bibitem[PB21]{pathria2021statistical}
R.K. Pathria and P.D. Beale.
\newblock {\em Statistical Mechanics}.
\newblock Elsevier, 2021.

\bibitem[Poi98]{poincare1898proprietes}
Henri Poincar{\'e}.
\newblock {\em Sur les propri{\'e}t{\'e}s du potentiel et sur les fonctions ab{\'e}liennes}.
\newblock 1898.

\bibitem[PP{\etalchar{+}}08]{petersen2008matrix}
Kaare~Brandt Petersen, Michael~Syskind Pedersen, et~al.
\newblock The matrix cookbook.
\newblock {\em Technical University of Denmark}, 7(15):510, 2008.

\bibitem[PXC{\etalchar{+}}21]{pu2021input}
Yi-Fei Pu, Xuetao Xie, Jinde Cao, Hua Chen, Kai Zhang, and Jian Wang.
\newblock An input weights dependent complex-valued learning algorithm based on {W}irtinger calculus.
\newblock {\em IEEE Transactions on Systems, Man, and Cybernetics: Systems}, 52(5):2920--2932, 2021.

\bibitem[Rem91]{remmert1991theory}
Reinhold Remmert.
\newblock {\em Theory of complex functions}, volume 122.
\newblock Springer Science \& Business Media, 1991.

\bibitem[Rud53]{rudin1953principles}
Walter Rudin.
\newblock {\em Principles of mathematical analysis}.
\newblock McGraw-Hill, Inc., 1953.

\bibitem[UARTND19]{utreras2019stochastic}
A~Utreras-Alarc{\'o}n, M~Rivera-Tapia, S~Niklitschek, and A~Delgado.
\newblock Stochastic optimization on complex variables and pure-state quantum tomography.
\newblock {\em Scientific Reports}, 9(1):16143, 2019.

\bibitem[VDB94]{van1994complex}
A~Van Den~Bos.
\newblock Complex gradient and {H}essian.
\newblock {\em IEE Proceedings-Vision, Image and Signal Processing}, 141(6):380--382, 1994.

\bibitem[Vir19]{virtue2019complex}
Patrick~M Virtue.
\newblock {\em Complex-valued deep learning with applications to magnetic resonance image synthesis}.
\newblock University of California, Berkeley, 2019.

\bibitem[Wil13]{wilde2013quantum}
Mark~M Wilde.
\newblock {\em Quantum information theory}.
\newblock Cambridge University Press, 2013.

\bibitem[Wir27]{wirtinger1927formalen}
Wilhelm Wirtinger.
\newblock Zur formalen {T}heorie der {F}unktionen von mehr komplexen {V}er{\"a}nderlichen.
\newblock {\em Mathematische Annalen}, 97(1):357--375, 1927.

\bibitem[ZLC{\etalchar{+}}19]{zhang2019fully}
Bingjie Zhang, Yusong Liu, Jinde Cao, Shujun Wu, and Jian Wang.
\newblock Fully complex conjugate gradient-based neural networks using {W}irtinger calculus framework: Deterministic convergence and its application.
\newblock {\em Neural Networks}, 115:50--64, 2019.

\end{thebibliography}

\end{document}